\documentclass[letterpaper,journal]{IEEEtran}
\usepackage{amsmath,amsfonts,amssymb}
\usepackage{algorithmic,algorithm,array,booktabs,graphicx,tabularx}
\usepackage[caption=false,font=footnotesize]{subfig}
\usepackage{textcomp,stfloats,url,cite}
\usepackage{tikz}
\usetikzlibrary{arrows.meta,positioning,calc}
\usepackage{balance}
\usepackage[hidelinks]{hyperref}
\hypersetup{pdftitle={Convex Inner Approximations for Safe AC Dispatch With Device Dynamics},pdfauthor={Ziqing Zhu},pdfsubject={IEEE journal manuscript},pdfkeywords={AC power flow, convex inner approximation, device dynamics, safety-critical control}}
\newtheorem{lemma}{Lemma}
\newtheorem{proposition}{Proposition}
\newtheorem{corollary}{Corollary}
\newcommand{\diag}{\operatorname{diag}}
\newcommand{\col}{\operatorname{col}}
\newcommand{\R}{\mathbb R}
\newcommand{\one}{\mathbf 1}
\newcommand{\Cac}{\mathcal C_{\mathrm{AC}}}
\newcommand{\Cex}{\mathcal C_{\mathrm{exec}}}
\newcommand{\Uad}{\mathcal U}
\newcommand{\norm}[1]{\left\lVert#1\right\rVert}
\newcommand{\LargeMedian}{1.82}
\newcommand{\LargePninetyfive}{3.84}
\newcommand{\LargeMax}{5.57}
\newcommand{\LargeMemory}{202.7}
\newcommand{\LargePercent}{9.3}
\newcommand{\LargeTrajectorySentence}{All 120 solutions satisfy the lifted-reference and AC checks; the independent instantaneous balancing-rate test admits 60. Four 300-/900-s held-reference executions, sampled every 2 s, satisfy AC, device-rate, balancing-input, and terminal-accuracy checks.}

\begin{document}
\title{Convex Inner Approximations for Safe AC Dispatch With Device Dynamics}
\author{Ziqing Zhu,~\IEEEmembership{Member,~IEEE}
\thanks{Corresponding author: Ziqing Zhu (e-mail: ziqzhu@polyu.edu.hk).}}
\markboth{Manuscript}{Zhu: Convex Inner Approximations for Safe AC Dispatch With Device Dynamics}
\maketitle
\begin{abstract}
Safe AC dispatch requires both feasible operating points and enough ramping capability to reach the next target. This paper constructs convex inner approximations that connect nonlinear network feasibility, constrained device response, and dispatch-interval tracking. A fixed-point certificate enforces AC operating limits under bounded load uncertainty. Componentwise response envelopes convert heterogeneous device dynamics into affine conditions on measured states and power references. A command governor then combines these conditions with available ramping capability and a hard terminal-error bound. It selects an admissible target through a minimum-deviation objective and holds the corresponding reference while reserved fast regulation remains active. The resulting quadratic program optimizes controllable references and auxiliary variables. Tests on the IEEE 39- and 118-bus systems include 96 single-target executions and 64 transitions in continuous dispatch sequences at five- and fifteen-minute intervals. All 160 transitions satisfy the AC constraints, independent balancing-generator capability checks, and the prescribed regulation-adjusted terminal accuracy. Complete admission-cycle medians are 3.45 and 8.79 ms. A connected synthetic system with 3,200 buses and 1,000 generating units extends the computational assessment to 999 independent references, with a median cycle time of \LargeMedian{} s. The construction provides a tractable link between available ramping capability and safe, timely execution of AC dispatch.
\end{abstract}

\section{Introduction}
\IEEEPARstart{P}{ower} dispatch determines operating targets, whereas the network experiences the injections produced as physical devices respond. Different response speeds and finite ramp rates make the transition between targets an essential part of dispatch feasibility. Two AC-feasible endpoints can be connected by an unsafe physical response. Equally, a safe response can approach the target too slowly for the next dispatch update. The dispatch interface must therefore relate the new target to the actual current state, the available ramping capability, and the time allowed for execution.

AC network constraints make this relationship difficult to compute. Voltages, reactive outputs, and branch apparent powers are coupled through nonlinear equations. Linear approximations and convex relaxations supply valuable computational tools \cite{molzahn,low,matpower}, while a safe operating range requires sufficient conditions for the original constraints. Reusable inner approximations offer such a description. Incorporating device response into that description supports dispatch admission and execution through compact reference optimization.

\subsection{Related Work and Methodological Position}
Fixed-point analysis establishes power-flow solvability and local uniqueness \cite{bolognani}. Convex restriction develops this principle into sufficient conditions for operating-limit satisfaction \cite{lee2019}, and robust convex restriction incorporates bounded injection uncertainty \cite{leerobust}. Robust and chance-constrained AC optimization address related uncertainty questions through admissible policies or probability distributions \cite{louca,roald,lorca}. Sequential convex restriction also identifies feasible operating paths \cite{leepath}. The present construction uses a fixed local AC certificate and adds the response conditions needed to execute a dispatch transition.

Ramping-capability formulations explicitly relate upward and downward movements to the available generation range and ramp rates \cite{navid,caisofrp}. This inter-interval requirement is distinct from regulation within a dispatch interval. Operating-reserve analysis assigns the latter to correction of short-term deviations around the scheduled response \cite{elareserve}. These established distinctions motivate a dispatch-time assessment of the newly requested movement, with a fixed allowance for faster regulation during execution.

Reference and command governors provide a mature framework for selecting admissible commands from a desired setpoint and the current state \cite{gilbert,bemporad,garone}. Predictive safety filters extend this approach to constrained nonlinear systems \cite{wabersich}. A command governor can optimize the reference vector directly while preserving a pre-existing device controller. For AC dispatch, the key is to express network safety and deadline-dependent device capability as convex reference conditions. This paper achieves that connection through a first-order response enclosure and a finite-time error bound.

Control barrier functions (CBFs) yield affine input conditions for preserving state constraints \cite{ames2017,ames2019}; higher-order constructions cover higher relative degrees \cite{xiao}. Sampled-data and input-delay formulations account for the evolution between control updates \cite{breeden,singletary,taylor}. Here, a performance-equivalent power-response model \cite{zhou} gives network-margin outputs of relative degree one. An explicit held-response enclosure covers the complete intersample evolution, alongside a first-order CBF condition. This combination supplies reusable admissibility conditions for the dispatch-time command governor.

\subsection{Contributions}
First, the paper constructs a convex admissible state set by combining robust AC certification with heterogeneous device-response bounds. Network and device constraints become affine inequalities in the measured state, reference, and auxiliary variables.

Second, it incorporates available ramping capability and terminal tracking accuracy into an admissible target set. A command governor selects the target at each dispatch update using the actual measured state. A hard endpoint condition guarantees the scheduled response reaches that target within the dispatch interval, while the held-response certificate covers continuous network safety and reserved regulation.

Third, it evaluates the construction on the IEEE 39- and 118-bus systems using paired solver-generated targets, five- and fifteen-minute dispatch intervals, continuous target reversals, and a terminal-condition ablation. Independent AC and balancing-capability checks connect the convex conditions to actual power trajectories. The experiments quantify target correction, tracking, and complete computational cost; a connected 3,200-bus synthetic benchmark additionally tests 999 independent references.

\section{AC Network and Device Response Models}
\subsection{AC Network Constraints}
Consider a fixed network topology and generator commitment. Let $V_i$ denote the complex voltage at bus $i$, $Y$ the bus-admittance matrix, and $P_i^{\rm G},Q_i^{\rm G},P_i^{\rm L},Q_i^{\rm L}$ the generation and load powers. The AC equations are
\begin{equation}
 P_i^{\rm G}-P_i^{\rm L}+\mathrm j(Q_i^{\rm G}-Q_i^{\rm L})
 =V_i\overline{(YV)_i} \label{eq:ac}
\end{equation}
The constraints include bus-voltage bounds, generator active- and reactive-power bounds, apparent-power limits at both ends of every branch, and branch angle-difference limits. The reference-bus voltage and PV-bus voltage magnitudes are fixed within an operating mode. The balancing generator has finite active- and reactive-power limits.

Let $p^0\in\R^{n_c}$ collect the active-power outputs of the independently controlled generators at a feasible base point. Define $u=p-p^0$, where $p$ is their actual active-power injection vector. Dependent network states $x\in\R^{n_x}$ contain the real and imaginary voltages at non-reference buses. Bounded factors $w$ describe load uncertainty. A fixed participation matrix maps these factors to buswise multiplicative errors shared by active and reactive demand. The identity mapping gives independent bus errors. Thus
\begin{equation}
 F(x,u,w)=0,\quad g(x,u,w)\le0,\quad |w|\le\bar w \label{eq:network}
\end{equation}
All vector inequalities and absolute values are componentwise. Powers are expressed in per unit unless a physical unit is stated. The objective is a sufficient description of injections for which \eqref{eq:network} is satisfied for every allowed $w$.

The network is represented by quasi-steady AC equations coupled to active-power response dynamics. Reactive outputs and balancing power are obtained from those equations. Their operating limits enter $g$; the balancing generator's response capability is additionally evaluated from the required power trajectory in Section IV. The resulting model couples quasi-steady network security to the explicitly represented active-power dynamics.

\subsection{Device Response and Reference Updates}
Let $\hat p^{\rm dis}$ denote a requested dispatch target, $p^{\rm tar}$ its accepted target, and $r$ the deviation of the device power reference from $p^0$. These are distinct from the actual injection $p=p^0+u$. Between reference updates, each independently controlled device obeys
\begin{equation}
 \begin{gathered}
 \dot u_i(t)=\frac{r_i+\delta_i(t)-u_i(t)}{\tau_i(t)},\\
 \underline\tau_i\le\tau_i(t)\le\overline\tau_i,\qquad |\delta_i(t)|\le a_i
 \end{gathered} \label{eq:device}
\end{equation}
Decompose $\delta=\delta^{\rm reg}+\delta^{\rm mod}$, with $|\delta^{\rm reg}|\le a^{\rm reg}$ and $|\delta^{\rm mod}|\le a^{\rm mod}$, so $a=a^{\rm reg}+a^{\rm mod}$. The two terms represent reserved fast regulation and input-equivalent modeling error. The measured deviation $\hat u$ satisfies $|u(0)-\hat u|\le\eta$. Power-reference and rate limits use $\underline u=p^{\min}-p^0$, $\overline u=p^{\max}-p^0$, and $|\dot u_i|\le R_i$, with $p^{\max}=\col_i\{P_i^{\max}\}$.

Dispatch updates occur at $T_k$ with $T_{\rm d}=T_{k+1}-T_k$. The tests use 300 and 900 s, matching the five- and fifteen-minute market intervals in CAISO's ramping-product deployment \cite{caiso2020}. The governor selects one scheduled reference for that interval; fast regulation changes every $T_{\rm f}$ seconds. CAISO distinguishes inter-interval ramping from post-dispatch deviations handled by regulation \cite[pp.~5--7]{caisofrp}. NREL likewise describes units holding or approaching their next economic setpoint while regulation corrects the remaining imbalance \cite[p.~14]{elareserve}. Accordingly, substantial changes in the scheduled operating point enter at dispatch updates. Modeled intra-interval variation remains within $a^{\rm reg}$ and $a^{\rm mod}$. Internal control steps track the accepted target and enforce the allocated regulation bounds; subsequent ramping requirements are reassessed at the next dispatch update. Network topology, commitment, and generator voltage setpoints remain fixed within each certified operating interval.

\subsection{Admissible States and References}
The construction begins with a robust AC inner approximation $\Cac$ in $u$ coordinates. For a measured state $\hat u$, an admissible reference must keep every response of \eqref{eq:device}, consistent with $\eta$, $a$, and the time-constant bounds, inside $\Cac$, while respecting the device limits. We seek a sufficient admissible reference set $\Uad(\hat u)$ and define
\begin{equation}
 \Cex=\{\hat u:\Uad(\hat u)\ne\varnothing\} \label{eq:admissible_intro}
\end{equation}
This admissible \emph{state} set is defined by the hard feasibility conditions; the requested target determines the preferred admissible reference. The deadline-dependent target set developed below further requires a prescribed terminal accuracy.

\section{Construction of a Convex Inner Approximation With Device Dynamics}
\subsection{Fixed-Point Certification of AC Feasibility}
\subsubsection{Exact quadratic representation}
Choose a feasible base voltage $x^0$ and let $v=x-x^0$. In rectangular coordinates, the primitive network quantities are exact quadratics:
\begin{equation}
 y_\alpha(x^0+v)=y_\alpha^0+G_\alpha v+v^\top Q_\alpha v \label{eq:quadratic}
\end{equation}
The primitives include nodal injections, squared voltage magnitudes, branch active and reactive powers, and the real and imaginary parts of cross-voltage products. Consequently, the square power-flow subsystem can be written as
\begin{equation}
 F_0+Jv+q_F(v)-B_u u-B_w w=0 \label{eq:expansion}
\end{equation}
where the entries of $q_F$ are quadratic forms. The subsystem includes active-power balances at non-reference buses, reactive-power balances at PQ buses, and squared-voltage equations at PV buses. Fixed-point convex restriction provides the underlying certification principle \cite{lee2019,leerobust}; the following formulation retains the exact rectangular-coordinate quadratic structure.

\subsubsection{Enclosing the AC solution}
Let $R_J$ be a nonsingular approximate inverse of $J$, retain $E=I-R_JJ$, and compute $S,T$ from the linear systems with right-hand sides $B_u,B_w$. For a prescribed injection box $|u|\le d$, define
\begin{equation}
 \begin{split}
 M&=[S\diag(d)\quad T\diag(\bar w)],\\
 B&=[B_u\diag(d)\quad B_w\diag(\bar w)],\\
 v&=M\xi+e,\qquad |\xi|\le\one,\quad |e|\le\bar e
 \end{split} \label{eq:enclosure}
\end{equation}
The normalized vector $\xi$ represents the actual $u,w$ within their boxes. The remaining voltage error is enclosed through the fixed-point map
\begin{equation}
 \Phi(e;\xi)=Ee-R_JF_0-R_J(JM-B)\xi-R_Jq_F(M\xi+e) \label{eq:map}
\end{equation}
Both inverse and linear-solve residuals are retained. In particular, the inverse is applied to the quadratic equations \emph{before} absolute values are taken, preserving cancellations between the AC equations.

For a symmetric matrix $Q$, introduce the computable envelope
\begin{equation}
 \begin{split}
 \beta(Q)=&\ \one^\top|M^\top QM|\one
 +2\bar e^\top|QM|\one+\bar e^\top|Q|\bar e
 \end{split} \label{eq:beta}
\end{equation}
Then $|(M\xi+e)^\top Q(M\xi+e)|\le\beta(Q)$ throughout \eqref{eq:enclosure}. For large sparse networks, the first term is bounded by $(|M|\one)^\top|QM|\one$ and evaluated rowwise in the sparse quadratic-monomial basis. This preserves sufficiency while avoiding a dense third-order coefficient tensor. Let $\widetilde Q_j=\sum_i(R_J)_{ji}Q_{F,i}$. A sufficient self-mapping condition is
\begin{equation}
 |E|\bar e+|R_JF_0|+|R_J(JM-B)|\one
 +\col_j\{\beta(\widetilde Q_j)\}\le\bar e \label{eq:selfmap}
\end{equation}
Continuity and Brouwer's theorem then give a root of \eqref{eq:expansion} in the enclosure for every allowed $u,w$. The implementation additionally checks a weighted infinity-norm contraction bound for the enclosed root. If the nonnegative Jacobian bound $K$ admits $z>0$ with $Kz<z$, the enclosed root is unique locally. This identifies a consistent AC solution branch as injections vary.

\begin{figure}[t]
\centering
\begin{tikzpicture}[font=\footnotesize,>=Latex,
 box/.style={draw=blue!40!black,fill=blue!3,rounded corners=1.5pt,align=center,text width=7.65cm,minimum height=9mm},
 note/.style={align=center,text=black!70}]
\node[box] (ac) {\textbf{AC inner approximation}\\Fixed-point enclosure and operating limits: $Au\le b$};
\node[box,below=6mm of ac] (resp) {\textbf{Device-response enclosure}\\Heterogeneous time constants, input errors, and ramp limits};
\node[box,below=6mm of resp] (set) {\textbf{Admissible state and reference sets}\\Affine lifted conditions in $(\hat u,r,s)$; projection onto $\hat u$};
\node[box,below=6mm of set] (qp) {\textbf{Dispatch-time command governor}\\Ramping capability and hard terminal accuracy};
\draw[->] (ac)--(resp);\draw[->] (resp)--(set);\draw[->] (set)--(qp);
\end{tikzpicture}
\caption{Construction and use of the inner approximation. The AC certificate and response enclosure define admissible states and references. Ramping and terminal conditions then determine the target that can be executed before the next dispatch update.}
\label{fig:method}
\end{figure}
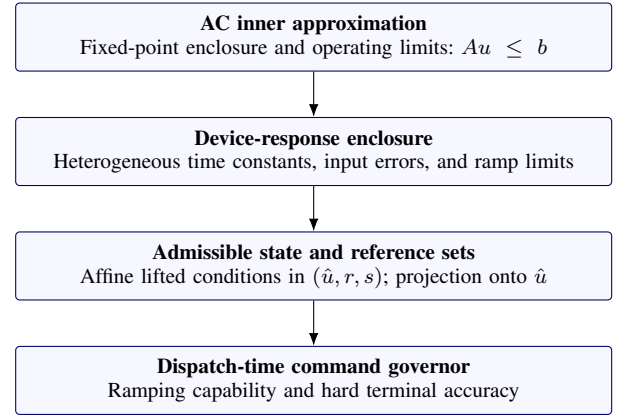

\subsubsection{Enforcing the original operating limits}
Write the primitive operating inequalities as $Ly(x)+D_w w\le h$. The first-order voltage error is eliminated using \eqref{eq:expansion}, giving
\begin{equation}
 \begin{split}
 y={}&y^0-GR_JF_0+G\widetilde S u+G\widetilde T w+GEe\\
 &+\col_\alpha\{v^\top Q_\alpha^{\rm eff}v\},\\
 \widetilde S={}&S-R_J(JS-B_u),\quad
 \widetilde T=T-R_J(JT-B_w),\\
 Q_\alpha^{\rm eff}={}&Q_\alpha-\sum_jG_{\alpha j}\widetilde Q_j
 \end{split} \label{eq:outputs}
\end{equation}
This identity combines the network equations and output bounds before taking their worst-case magnitudes. With $\beta_y=\col_\alpha\{\beta(Q_\alpha^{\rm eff})\}$, sufficient output conditions are
\begin{equation}
 \begin{split}
 A_o={}&LG\widetilde S,\\
 b_o={}&h-L(y^0-GR_JF_0)-|LG\widetilde T+D_w|\bar w\\
 &-|LGE|\bar e-|L|\beta_y,\qquad A_ou\le b_o
 \end{split} \label{eq:outputbound}
\end{equation}
Append the exact controlled-generator active-power bounds and $|u|\le d$ to obtain
\begin{equation}
 \Cac=\{u:Au\le b\} \label{eq:CAC}
\end{equation}
Rowwise support bounds remove inequalities that are redundant over the complete box $|u|\le d$, reducing the online constraint count while preserving all certified operating limits.

At each branch end, the apparent-power limit is enforced through a 32-facet \emph{inscribed} polygon:
\begin{equation}
 \cos\theta_k P+\sin\theta_k Q
 \le S^{\max}\cos(\pi/32),\qquad \theta_k=2\pi k/32 \label{eq:polygon}
\end{equation}
The cosine factor ensures inclusion inside the original circle. Voltage-magnitude limits use the squared-voltage primitives. Angle limits use the real and imaginary cross-voltage products together with a positive-real-product condition. Reactive-power bounds include the uncertain reactive load at each generator bus. Therefore, \eqref{eq:CAC} certifies operating limits as well as the existence of an AC solution.

\subsection{Device-Response Bounds During Reference Holding}
Heterogeneous devices follow different response curves. For constant time constants and zero disturbance, $u_i(t)=e^{-t/\tau_i}u_i(0)+(1-e^{-t/\tau_i})r_i$. The coefficients differ between devices. Hence the vector $u(t)$ generally does not lie on the line segment joining $u(0)$ and $r$, even if both endpoints belong to $\Cac$.

Let elapsed time after a reference update be partitioned into cells $I_\ell=[t_\ell,t_{\ell+1}]$. All exponential expressions below are componentwise. Define
\begin{equation}
 \begin{split}
 E_\ell^-&=\exp(-t_{\ell+1}/\underline\tau),\quad
 E_\ell^+=\exp(-t_\ell/\overline\tau),\\
 E_\ell^{\rm m}&=(E_\ell^++E_\ell^-)/2,\quad
 E_\ell^{\rm r}=(E_\ell^+-E_\ell^-)/2
 \end{split} \label{eq:exponents}
\end{equation}
For the terminal cell $[t_L,\infty)$, $E_L^-=0$.

\begin{lemma}[Reference-holding enclosure]
For \eqref{eq:device}, a constant reference $r$, and $|u(0)-\hat u|\le\eta$, every response in $I_\ell$ satisfies $|u(t)-c_\ell|\le\rho_\ell$, where
\begin{equation}
 \begin{split}
 c_\ell&=E_\ell^{\rm m}\odot\hat u+(\one-E_\ell^{\rm m})\odot r,\\
 \rho_\ell&=E_\ell^{\rm r}\odot|r-\hat u|
 +E_\ell^+\odot\eta+(\one-E_\ell^-)\odot a
 \end{split} \label{eq:envelope}
\end{equation}
Here $\odot$ denotes componentwise multiplication.
\end{lemma}
\begin{IEEEproof}
Set $k_i(t)=1/\tau_i(t)$ and $E_i(t)=\exp[-\int_0^t k_i(q)dq]$. The scalar variation-of-constants formula gives
\begin{equation}
 \begin{split}
 u_i(t)={}&E_i(t)u_i(0)+(1-E_i(t))r_i\\
 &+\int_0^t e^{-\int_s^t k_i(q)dq}k_i(s)\delta_i(s)ds
 \end{split} \label{eq:variation}
\end{equation}
The kernel is nonnegative and integrates to $1-E_i(t)$. Thus the last term has magnitude at most $(1-E_i(t))a_i$. On $I_\ell$, $E_{\ell,i}^-\le E_i(t)\le E_{\ell,i}^+$. Writing $u_i(0)=\hat u_i+e_i^0$, the deviation from the stated center is the sum of $(E_i-E_{\ell,i}^{\rm m})(\hat u_i-r_i)$, $E_ie_i^0$, and the integral term. Bounding these three terms yields \eqref{eq:envelope}. The argument permits time-varying $\tau_i$ and arbitrary bounded $\delta_i$.
\end{IEEEproof}

Introduce $s\ge|r-\hat u|$. Then the support function of the interval box yields the affine sufficient condition
\begin{equation}
 Ac_\ell+|A|\left[E_\ell^{\rm r}\odot s+
 E_\ell^+\odot\eta+(\one-E_\ell^-)\odot a\right]\le b \label{eq:hold}
\end{equation}
Unlike a pointwise power-flow test, this condition covers all times within a cell. A finite collection of cells plus the terminal cell covers every possible reference-holding duration. The terminal cell also enforces the limiting reference neighborhood, so an accepted reference can remain in use if the next tracking update is unavailable.

\subsection{Convex Construction of the Admissible State Set}
\subsubsection{Hard device and network conditions}
The initial measurement box is required to lie inside $\Cac$:
\begin{equation}
 A\hat u+|A|\eta\le b \label{eq:initial}
\end{equation}
The following affine device constraints bound both the effective input and its resulting rate of change:
\begin{equation}
 \begin{split}
 \underline u+a&\le r\le\overline u-a,\qquad s\ge r-\hat u,\quad s\ge\hat u-r,\\
 s+\eta+a&\le\underline\tau\odot R,\qquad
 2a\le\underline\tau\odot R 
 \end{split} \label{eq:actuator}
\end{equation}
For a held reference, positivity gives $|u_i(t)-r_i|\le\max\{|u_i(0)-r_i|,a_i\}$. Therefore the numerator of \eqref{eq:device} has magnitude at most $\max\{s_i+\eta_i,a_i\}+a_i$, which is bounded by \eqref{eq:actuator}. This bounds the continuous derivative throughout the hold.

We also impose a robust first-order CBF condition at an update. Define $h_j(u)=b_j-A_ju$, $k^{\rm m}=(1/\underline\tau+1/\overline\tau)/2$, and $k^{\rm r}=(1/\underline\tau-1/\overline\tau)/2$. For $\alpha>0$, a sufficient condition for $\dot h_j\ge-\alpha h_j$ over the current measurement box is
\begin{equation}
 \begin{aligned}
 &A\diag(k^{\rm m})(r-\hat u)+|A|\diag(k^{\rm r})s\\
 &\quad+|A|\diag(1/\underline\tau)(\eta+a)\\
 &\qquad\le\alpha(b-A\hat u-|A|\eta)
 \end{aligned} \label{eq:cbf}
\end{equation}
The affine margins have relative degree one with respect to $r$ under \eqref{eq:device}. Thus \eqref{eq:cbf} is the appropriate CBF specialization \cite{ames2017,xiao}. It is used alongside \eqref{eq:hold}: the CBF controls the local direction, while the response enclosure covers the sampled-data evolution \cite{breeden}.

\subsubsection{Lifted polyhedron and its projection}
Collect \eqref{eq:hold}--\eqref{eq:cbf} into a polyhedron
\begin{equation}
 \begin{split}
 \mathcal P&=\{(\hat u,r,s):H_u\hat u+H_rr+H_ss\le q\},\\
 \Uad(\hat u)&=\{r:\exists s,\ (\hat u,r,s)\in\mathcal P\},\\
 \Cex&=\{\hat u:\exists r,s,\ (\hat u,r,s)\in\mathcal P\}
 \end{split} \label{eq:polyhedron}
\end{equation}
The matrices are constant for fixed network, device, and uncertainty bounds. No explicit elimination of $r,s$ is required to test membership: feasibility of the lifted inequalities suffices. Projection preserves convexity \cite{boyd}, making \eqref{eq:polyhedron} a useful state-space description as well as an online computational representation.

\begin{proposition}[Convex admissibility and constraint satisfaction]
Suppose \eqref{eq:selfmap} holds, \eqref{eq:outputbound} and the domain bounds define $\Cac$, and the time cells cover $[0,\infty)$. Then $\Cex$ is a convex polyhedron. For every $\hat u\in\Cex$, any witnessing reference $r\in\Uad(\hat u)$ keeps all responses of \eqref{eq:device} inside $\Cac$ and satisfies the device input and rate limits throughout its hold, for every initial state in the measurement box and every permitted disturbance and time constant. At each resulting injection and allowed load realization, an AC solution in the certified enclosure satisfies the original operating limits.
\end{proposition}
\begin{IEEEproof}
For fixed bounds, the centers in \eqref{eq:envelope} are affine in $(\hat u,r)$ and the radii in \eqref{eq:hold} are affine in $s$. All remaining constraints are affine. Hence $\mathcal P$ and its projection $\Cex$ are convex polyhedra. Given a witness, the lemma and the support bound \eqref{eq:hold} imply $Au(t)\le b$ on every cell, including the tail. Equation \eqref{eq:initial} covers the initial measurement box, and \eqref{eq:actuator} gives the input and continuous-rate bounds. Every resulting $u(t)$ remains in the domain for which \eqref{eq:selfmap} supplies an AC solution. Substitution of the envelope bounds into \eqref{eq:outputs} gives \eqref{eq:outputbound}, so that solution satisfies the original operating inequalities. If the contraction check is also satisfied, the enclosed solution is the unique local branch. The guarantee follows from the accepted state--reference pair and remains valid throughout that reference hold.
\end{IEEEproof}

A simple quantitative nonemptiness test fixes $r=0$ and seeks the largest $\rho\ge0$ for which $|\hat u_i|\le\rho P_i^{\max}$ satisfies all lifted conditions with $s_i=\rho P_i^{\max}$. Taking rowwise supports reduces the test to scalar inequalities in $\rho$. This yields an explicit centered box inside $\Cex$. General membership in \eqref{eq:polyhedron} additionally permits other reference choices. The distinction is important when comparing a high-dimensional box with a two-dimensional section.

\subsection{Ramping Capability and Reference Selection}
\subsubsection{Available movement at a dispatch update}
The initial state is $u(T_k^-)$, including the response remaining from the preceding interval. The governor uses its measurement $\hat u_k$, preserving any deviation from the preceding target. Conventional ramping-capability screening combines generation headroom and rate capability over the dispatch interval \cite{navid,caisofrp}. With the reserved input range already imposed in \eqref{eq:actuator}, the robust inter-interval movement condition is
\begin{equation}
 -R T_{\rm d}+\eta\le u^{\rm tar}-\hat u_k
 \le R T_{\rm d}-\eta \label{eq:interramp}
\end{equation}
For example, the available upward movement is the smaller of $R_iT_{\rm d}-\eta_i$ and $\overline u_i-a_i-\hat u_{k,i}$; the downward expression uses $\hat u_{k,i}-\underline u_i-a_i$. The continuous-rate constraints in \eqref{eq:actuator} additionally account for device lag throughout execution. Thus satisfying an aggregate energy or endpoint ramp limit alone is insufficient for admission.

\subsubsection{Terminal accuracy under reserved regulation}
Set the held reference to the accepted target, $r=u^{\rm tar}=p^{\rm tar}-p^0$. Over the current interval, define the regulation response $z_i^{\rm reg}$ by $\dot z_i^{\rm reg}=(\delta_i^{\rm reg}-z_i^{\rm reg})/\tau_i(t)$, $z_i^{\rm reg}(0)=0$. Each interval inherits the actual initial power in $u(0)$. Then $u-z^{\rm reg}$ follows \eqref{eq:device} with only $\delta^{\rm mod}$. For the prescribed tolerance $\varepsilon>0$, impose
\begin{equation}
 \begin{gathered}
 e_i^+=e^{-T_{\rm d}/\overline\tau_i},\qquad
 e_i^-=e^{-T_{\rm d}/\underline\tau_i},\\
 e_i^+(s_i+\eta_i)+(1-e_i^-)a_i^{\rm mod}
 \le\varepsilon_i \label{eq:terminal}
 \end{gathered}
\end{equation}
The regulation response is deliberately separated from tracking: a nonzero regulation request changes actual power even when the scheduled target has been attained. Both components remain covered by the network and continuous-rate conditions through the full allowance $a$.

\begin{corollary}[Safe dispatch-interval target attainment]
Let $r=u^{\rm tar}$ satisfy \eqref{eq:polyhedron}, \eqref{eq:interramp}, and \eqref{eq:terminal}. Under the bounds of Proposition 1, the held response satisfies the certified network and device constraints, and
\begin{equation}
 \begin{aligned}
 |u(T_{\rm d})-z^{\rm reg}(T_{\rm d})-u^{\rm tar}|&\le\varepsilon,\\
 |u(T_{\rm d})-u^{\rm tar}|&\le\varepsilon+a^{\rm reg}
 \end{aligned} \label{eq:attainment}
\end{equation}
\end{corollary}
\begin{IEEEproof}
Apply \eqref{eq:variation} to $u-z^{\rm reg}-u^{\rm tar}$. The initial deviation is bounded by $s+\eta$, its transition coefficient by $e^+$, and the modeling-error integral by $(1-e^-)a^{\rm mod}$. Equation \eqref{eq:terminal} yields the first inequality. Positivity gives $|z^{\rm reg}|\le a^{\rm reg}$, yielding the second. Proposition 1 supplies the held-response safety guarantee.
\end{IEEEproof}

\subsubsection{Convex command governor}
Denote by $\mathcal T(\hat u_k,T_{\rm d})$ the targets admitting an $s$ that satisfies these conditions, together with the steady-reference guard $Au^{\rm tar}+|A|(a+\eta)\le b$. All inequalities are affine for a given interval, so this target set is a convex polyhedron. Following the command-governor formulation \cite{garone}, select
\begin{equation}
 \begin{aligned}
 \min_{u^{\rm tar},s}\quad &\tfrac12\norm{u^{\rm tar}-\hat u^{\rm dis}}_W^2
                       +\tfrac{\epsilon_s}{2}\norm{s}_W^2\\
 \mathrm{s.t.}\quad &(\hat u_k,u^{\rm tar},s)\in\mathcal P,\\
 &\eqref{eq:interramp},\ \eqref{eq:terminal},\quad
 Au^{\rm tar}+|A|(a+\eta)\le b
 \end{aligned} \label{eq:qp}
\end{equation}
where $\hat u^{\rm dis}=\hat p^{\rm dis}-p^0$, $W=\diag(d_i^{-2})$, and $\epsilon_s=10^{-8}$ regularizes the auxiliary variables. The objective minimizes deviation within the fixed admissible target set. Safety and terminal accuracy are hard constraints. The formulation has $2n_c$ decision variables, independently of the number of held-response cells or AC network states.

\begin{algorithm}[t]
\caption{Dispatch-time command governor with a fixed AC certificate}
\begin{algorithmic}[1]
\STATE Offline: compute the fixed-point and operating-limit bounds.
\STATE Construct $A,b$, device bounds, and held-response cells.
\STATE At dispatch update $T_k$: validate the operating mode.
\STATE Obtain $\hat u_k$ from the actual current power state.
\STATE Form the headroom, ramp, and terminal conditions for $T_{\rm d}$.
\STATE Solve \eqref{eq:qp} for the accepted target and reference.
\STATE Check primal residuals, the hold envelope, and terminal bound.
\IF{the reference passes these checks}
\STATE Apply and store $r=u^{\rm tar}$ and its holding certificate.
\ELSE
\STATE Retain the stored reference while its bounds remain valid.
\ENDIF
\STATE Track the accepted target with reserved fast regulation.
\end{algorithmic}
\end{algorithm}

Network factorizations, quadratic-form bounds, and $A,b$ are reused offline objects. Online computation updates the measured state, requested target, and interval-dependent inequalities. Topology or operating-limit changes trigger certificate revalidation. A stored reference retains its all-time holding certificate under the stated bounds and supplies an admissible hold between accepted updates. Algorithm 1 makes the inter-interval decision at $T_k$; internal control steps retain the allocated regulation band and the accepted scheduled reference.

\section{Case Studies}
\subsection{Simulation Setup}
The tests use the PGLib-OPF v23.07 versions of the IEEE 39- and 118-bus systems \cite{pglib}. The 35 zero-active-capacity synchronous condensers in the 118-bus model retain their reactive-power role. The independently controlled active-power devices number 9 and 18, respectively. Original voltage, generator, branch apparent-power, and angle limits are retained. AC-feasible base points are computed with IPOPT \cite{ipopt}. The fixed certificates at 90\% and 100\% load characterize the admissible regions; the dispatch-execution experiments use the nominal-load certificates throughout.

\begin{table}[t]
\caption{Test Systems and Simulation Parameters}\label{tab:setup}
\centering\footnotesize
\begin{tabular}{lrr}
\toprule
Parameter & 39-bus & 118-bus\\\midrule
Buses / branches & 39 / 46 & 118 / 186\\
Independent active-power controls $n_c$ & 9 & 18\\
Independent load errors & 21 & 99\\
Balancing-generator $P^{\max}$ (MW) & 646 & 1182\\
Nominal certificate-domain half-width & 1.5\% & 2.0\%\\
Nominal device time constants (s) & \multicolumn{2}{c}{30--90}\\
Time-constant bounds & \multicolumn{2}{c}{$[0.8,1.2]$ times nominal}\\
Device ramp limit & \multicolumn{2}{c}{3\% $P^{\max}$/min}\\
Load-error bound $\bar w$ & \multicolumn{2}{c}{$\pm0.5\%$ per loaded bus}\\
Measurement-error bound $\eta$ & \multicolumn{2}{c}{$\pm0.025\%\ P^{\max}$}\\
Regulation / model-error allowance & \multicolumn{2}{c}{$\pm0.1\%$ / $\pm0.02\%\ P^{\max}$}\\
Dispatch / fast-regulation period (s) & \multicolumn{2}{c}{300 or 900 / 4}\\
Terminal tolerance $\varepsilon$ & \multicolumn{2}{c}{0.1\% $P^{\max}$}\\
Simulation output interval (s) & \multicolumn{2}{c}{0.5}\\
Balancing-generator time constant (s) & \multicolumn{2}{c}{30}\\\bottomrule
\end{tabular}
\par\smallskip\raggedright Network data and operating limits are from the benchmark. Device response, ramp rates, and error bounds are declared simulation parameters and remain unchanged across paired methods.
\end{table}

Table~\ref{tab:setup} gives the model settings. Load profiles combine common and bus-specific smooth components within the certified uncertainty bound. Regulation requests equal actual minus forecast demand, and fixed participation factors allocate each request within $a^{\rm reg}$. This allocation separates the forecast component of the dispatch target from its realized regulation requirement. The power-response simulation uses the exact scalar exponential update on each 0.5-s step. An exact rate-limited update applies when a comparison method reaches a device limit; these events are recorded separately.

The request bank contains 48 conventional fixed-commitment, N-0 network-constrained DC dispatch solutions. Four forecasts per system comprise uniform shifts of $-0.3\%$, zero, and $+0.3\%$, and a fixed spatial perturbation. Each is solved with affine losses, frozen losses, and penalized network-constraint slack, under two movement budgets. The small budget is $\pm1\%\ P^{\max}$; the large budget is the five-minute rate allowance, up to $\pm15\%\ P^{\max}$. Both dispatch durations use this same request bank, isolating the effect of available execution time. Requests are generated independently of the AC certificate. Comparisons are paired by forecast and model variant, with 12 requests per system and movement scale.

The proposed command governor is compared with \emph{static projection}, which applies the fixed-AC-region correction directly, and a \emph{CBF--CLF baseline}, which selects references every 60 s using the same hard held-response conditions and a soft quadratic control Lyapunov function. The baseline weights and 60-s CLF scaling are fixed. The governor instead selects an accepted target with the hard terminal condition and holds its reference for the dispatch interval. All methods retain the nonzero 4-s regulation channel, the same physical limits, and matched disturbances. Each method's correction and resulting terminal error are both reported.

For continuous operation, four successive requests alternate small and large upward and downward movements. Each accepted target is computed from the actual preceding endpoint, preserving continuous physical power states. Two time-varying response realizations and the constant lower and upper time-constant bounds give four sequences per system and dispatch duration. A separate paired ablation uses all 24 large requests with $\tau=\overline\tau$, enabling or disabling only \eqref{eq:terminal}.

The balancing generator is checked from the AC-required trajectory. Its derivative follows from differentiating the original AC equations, and its required first-order input is
\begin{equation}
 r_b^{\rm req}(t)=P_b^{\rm req}(t)+\tau_b\dot P_b^{\rm req}(t) \label{eq:balance}
\end{equation}
Power, input, and ramp bounds are tested independently of the controlled-device certificate. This evaluates the ability of a finite balancing actuator to realize the required AC trajectory. Regulation-adjusted terminal error is the maximum of $|u_i-z_i^{\rm reg}-u_i^{\rm tar}|/P_i^{\max}$; actual terminal deviation is also evaluated against \eqref{eq:attainment}. An independent implementation reconstructs the network and replays all device responses from the saved inputs.

\subsection{Geometry and Admissible Operating Range}
\begin{table}[t]
\caption{Centered Inner Boxes and Offline Computation}\label{tab:regions}
\centering\footnotesize
\begin{tabular}{llrrrr}
\toprule
System & Load & $\Cac$ & $\Cex$ subset & Faces & Offline (s)\\
 & & (\% $P^{\max}$) & (\% $P^{\max}$) & &\\\midrule
39-bus & 90\% & 1.1750 & 1.0550 & 25 & 0.280\\
39-bus & 100\% & 0.5732 & 0.5328 & 24 & 0.244\\
118-bus & 90\% & 0.4441 & 0.1902 & 38 & 4.031\\
118-bus & 100\% & 0.3220 & 0.0600 & 38 & 4.045\\\bottomrule
\end{tabular}
\par\smallskip\raggedright Entries are box half-widths. The $\Cex$ subset uses the witness $r=0$; general membership permits other references. Offline time includes local model preparation and construction of the selected certificate.
\end{table}

Table~\ref{tab:regions} establishes nonempty state neighborhoods at both load levels. At nominal load, the centered admissible-state half-widths are 0.5328\% and 0.0600\% $P^{\max}$. Device-response conditions reduce the static box because the initial measurement error and the full held response must fit within the network margins. This calculation supplies a simple guaranteed subset; state-specific reference selection uses the full lifted representation.

\begin{figure*}[t]
\centering
\subfloat[39-bus system.]{\includegraphics[width=.485\textwidth]{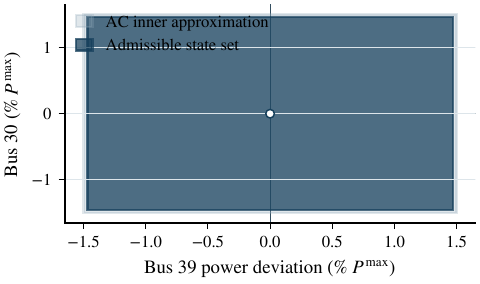}}\hfill
\subfloat[118-bus system.]{\includegraphics[width=.485\textwidth]{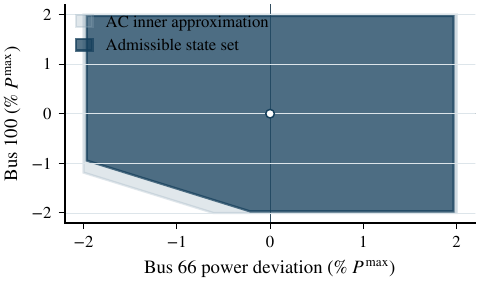}}
\caption{Nominal-load sections of the AC inner approximation and admissible state set. All other measured deviations are zero. Each polygon is the convex hull of feasible points obtained in 180 support directions; each admissible-state point has a feasible reference. The axes correspond to the two largest independently controlled generators. The section and the full-dimensional centered box quantify different aspects of the same lifted representation.}
\label{fig:geometry}
\end{figure*}

The sections in Fig.~\ref{fig:geometry} reveal the directional dependence of admissibility. A reference is optimized at each support point, retaining network coupling while allowing different devices to contribute to feasibility. The larger directional extents compared with the centered boxes demonstrate the value of choosing a reference for the measured state. Nonzero-reference tests further establish usable control freedom: 131 of 156 nonzero aggregate requests admit references when current nominal balancing-input and ramp conditions are included. The fast-regulation allowances total $\pm6.721$ MW and $\pm5.333$ MW in the two systems.

The deadline-dependent target set adds a second, physically different question: which next target can be reached from that state? In \eqref{eq:terminal}, a longer interval reduces the contribution of the initial target distance, whereas in \eqref{eq:actuator} a higher rate limit permits a larger initial response. The support calculations in Section~IV-D quantify both effects with the AC certificate held fixed.

\subsection{Execution of Solver-Generated Dispatch Targets}
\begin{table*}[t]
\caption{Paired Single-Target Dispatch Execution}\label{tab:execution}
\centering\footnotesize
\setlength{\tabcolsep}{7pt}
\begin{tabular}{llrrrrrrr}
\toprule
System & Method & Interval & AC limits & Balancing & No extra & Terminal & Maximum & Mean target\\
 & & (min) & satisfied & capability & limiting & accuracy & error (\%) & correction (MW)\\\midrule
39-bus & Static projection & 5 & 24/24 & 9/24 & 21/24 & 24/24 & 0.0516 & 127.56\\
 & CBF--CLF baseline & 5 & 24/24 & 24/24 & 24/24 & 9/24 & 0.3570 & 127.56\\
 & Command governor & 5 & 24/24 & 24/24 & 24/24 & 24/24 & 0.0474 & 135.71\\
 & Command governor & 15 & 24/24 & 24/24 & 24/24 & 24/24 & 0.0068 & 135.51\\\midrule
118-bus & Static projection & 5 & 24/24 & 24/24 & 12/24 & 24/24 & 0.0702 & 99.70\\
 & CBF--CLF baseline & 5 & 24/24 & 24/24 & 24/24 & 11/24 & 0.8210 & 99.70\\
 & Command governor & 5 & 24/24 & 24/24 & 24/24 & 24/24 & 0.0475 & 103.68\\
 & Command governor & 15 & 24/24 & 24/24 & 24/24 & 24/24 & 0.0082 & 103.29\\\bottomrule
\end{tabular}
\par\smallskip\raggedright Each row includes 12 small and 12 large requests. Terminal accuracy is measured relative to the method's accepted target after accounting for intentional regulation, with a 0.1\% $P^{\max}$ tolerance. Correction is the requested-to-accepted $\ell_2$ distance. Extra limiting records additional actuator saturation or rate limiting. Balancing capability independently checks \eqref{eq:balance}. Every regulation request in these comparisons is fulfilled.
\end{table*}

The command governor converts feasible held responses into timely dispatch transitions. In Table~\ref{tab:execution}, all 96 accepted single-target executions, covering both systems and dispatch durations, satisfy the AC limits, finite balancing capability, and terminal accuracy. No additional actuator limiting is activated. At five minutes, the largest regulation-adjusted errors are 0.0474\% and 0.0475\% $P^{\max}$; at fifteen minutes they fall to 0.0068\% and 0.0082\%. Thus the same network certificate supports both dispatch timescales through the interval-dependent target constraints.

The paired comparisons identify the roles of the different conditions. Static projection prioritizes target tracking through a static network-feasibility constraint. The CBF--CLF baseline preserves the held-response limits while treating convergence as a soft preference. The command governor imposes the deadline during target selection, achieving all three objectives together: AC feasibility, bounded device response, and terminal accuracy. The higher-dimensional target decision also allows the necessary correction to be distributed across devices instead of imposing a common progress factor.

\begin{figure*}[t]
\centering
\subfloat[39-bus: slowest controlled device, bus 39.]{\includegraphics[width=.485\textwidth]{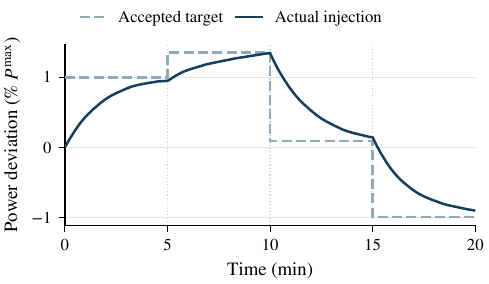}}\hfill
\subfloat[118-bus: slowest controlled device, bus 111.]{\includegraphics[width=.485\textwidth]{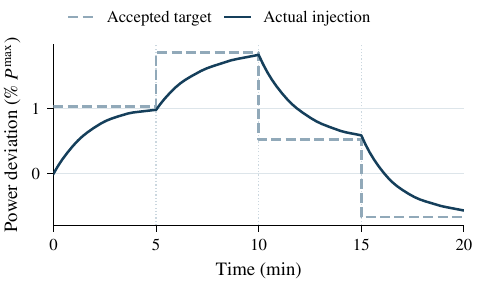}}\\
\subfloat[39-bus balancing-generator ramp requirement.]{\includegraphics[width=.485\textwidth]{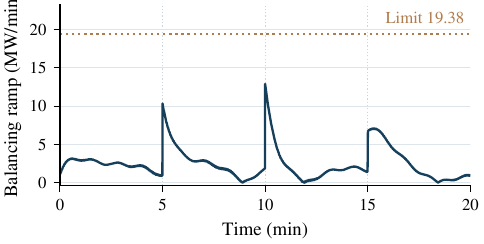}}\hfill
\subfloat[118-bus balancing-generator ramp requirement.]{\includegraphics[width=.485\textwidth]{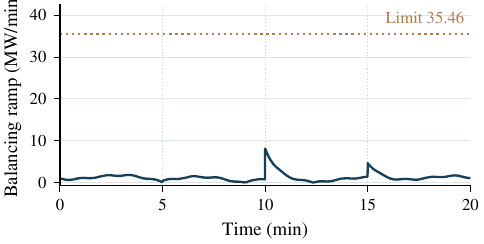}}
\caption{Continuous four-target execution with five-minute dispatch updates and all device time constants at their upper bounds. The preceding physical endpoint becomes the initial state of the next interval. Upper panels show accepted scheduled targets and actual injections; fast regulation remains active every 4 s. Lower panels show the magnitude of the AC-required balancing derivative and its unchanged physical limit.}
\label{fig:trajectories}
\end{figure*}

Fig.~\ref{fig:trajectories} illustrates continuous operation with target reversals and slow devices. Actual power remains continuous when the target changes, and the new governor calculation starts from that actual power. Across the 16 continuous sequences, all 64 transitions satisfy the AC constraints, balancing capability, and hard terminal accuracy. The largest regulation-adjusted endpoint error is 0.0819\% $P^{\max}$. The largest actual endpoint deviation is 0.1078\%, within the 0.2\% band accounting for both tracking tolerance and reserved regulation. These separate measurements preserve the physical contribution of the fast channel.

The accepted transitions produce nonzero changes in actual generation. Across the 160 primary governor transitions, the actual controlled-power movement has a mean $\ell_2$ magnitude of 19.96 MW and ranges from 5.57 to 47.47 MW. The largest balancing ramp in all primary governor trials is 18.70 MW/min for the 39-bus system and 22.70 MW/min for the 118-bus system, below the corresponding 19.38 and 35.46 MW/min limits.

\begin{figure*}[t]
\centering
\subfloat[Requested-to-accepted target correction.]{\includegraphics[width=.485\textwidth]{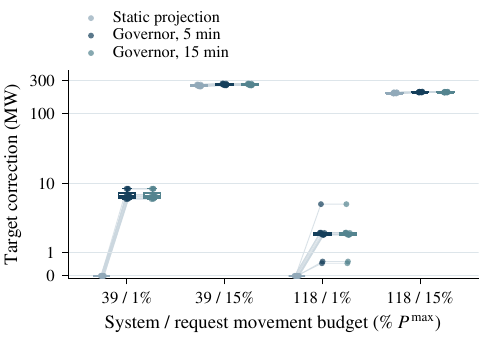}}\hfill
\subfloat[Regulation-adjusted terminal tracking.]{\includegraphics[width=.485\textwidth]{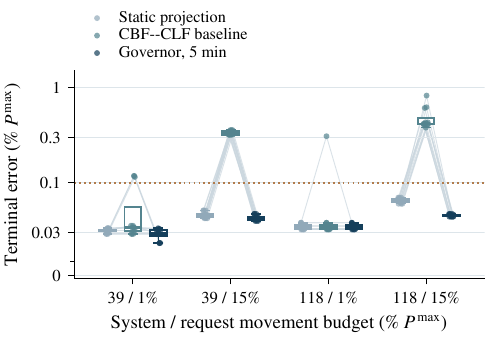}}
\caption{Paired single-target distributions, with all 12 cases retained in each system and request-scale group. Boxes show medians and interquartile ranges; scatter points show individual observations and thin connectors link matching requests. Corrections and tracking errors are reported together because each method can accept a different target. The dotted line in (b) is the hard terminal tolerance. Symmetric logarithmic scales retain zero corrections and small errors.}
\label{fig:distributions}
\end{figure*}

Fig.~\ref{fig:distributions} shows how the governor allocates the additional correction needed for execution. At five minutes, its mean correction exceeds static projection by 8.15 MW in the 39-bus system and 3.98 MW in the 118-bus system, while the accepted targets all meet the deadline. The target movement scale is kept visible: most correction in the large-request group comes from the fixed AC region, and the added device and terminal conditions refine the executable endpoint. Increasing the interval to fifteen minutes reduces the governor correction with the same network region and physical limits.

\subsection{Effects of Dispatch Interval and Response Uncertainty}
\begin{figure*}[t]
\centering
\subfloat[Terminal-condition ablation.]{\includegraphics[width=.325\textwidth]{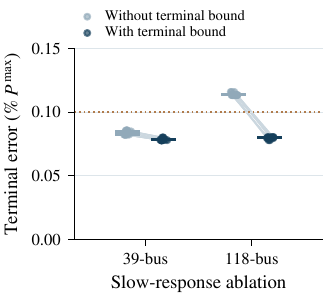}}\hfill
\subfloat[Paired continuous sequences.]{\includegraphics[width=.325\textwidth]{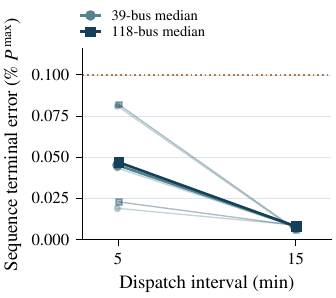}}\hfill
\subfloat[Admissible upward target range.]{\includegraphics[width=.325\textwidth]{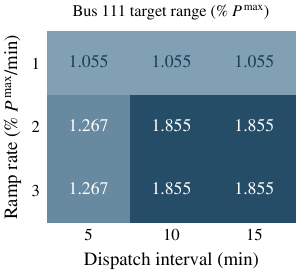}}
\caption{Dispatch-time capability and terminal accuracy. (a) All large requests are replayed with $\tau=\overline\tau$, changing only the terminal constraint. (b) Each thin line connects one of four matched sequences per system; the plotted error is the maximum across its four dispatch boundaries. (c) Linear programs maximize the bus-111 target in $\mathcal T(0,T_{\rm d})$ on the 118-bus system, allowing other reference components to vary. Entries are support values in percent of $P^{\max}$, with the AC certificate, error bounds, and terminal tolerance unchanged.}
\label{fig:sensitivity}
\end{figure*}

The ablation in Fig.~\ref{fig:sensitivity}(a) isolates the contribution of the hard terminal bound. Both variants use the same convex response constraints, request, disturbance, and slow time constants. On the 118-bus system, enforcing \eqref{eq:terminal} increases the terminal-accuracy count from 0/12 to 12/12 and reduces the maximum error from 0.1160\% to 0.0805\% $P^{\max}$. Both variants satisfy the terminal tolerance for the 39-bus cases. All 24 executions with the terminal bound preserve the network and finite-device checks. The result identifies the terminal-response condition as the mechanism enforcing the dispatch deadline.

The continuous-sequence comparison in Fig.~\ref{fig:sensitivity}(b) includes two time-varying realizations and both constant response bounds. All five- and fifteen-minute sequences achieve their accepted targets. Increasing $T_{\rm d}$ reduces the maximum endpoint error from 0.0811\% to 0.0087\% for the 39-bus system and from 0.0819\% to 0.0084\% for the 118-bus system. Fast regulation remains active at the same four-second period in both cases. The improvement therefore reflects the scheduled response time, with internal variability treated through the unchanged regulation allowance.

Fig.~\ref{fig:sensitivity}(c) directly quantifies available ramping flexibility. At 1\% $P^{\max}$/min, the admissible upward bus-111 target is 1.055\% $P^{\max}$ for all three tested intervals. Raising the ramp rate to 2\% increases the five-minute value to 1.267\%, while ten and fifteen minutes permit 1.855\%. Further increasing the rate to 3\% leaves these support values unchanged. This saturation has a clear interpretation: once the local rate requirement is satisfied, the terminal response and fixed network margins determine the remaining range. The 39-bus support calculations show the same interaction, with the corresponding long-interval value of 1.355\% at bus 39.

These results distinguish dispatch-time ramping capability from continuous reserve rescheduling. The governor evaluates the new movement once at the boundary using the measured state and the available interval. The response certificate then covers the intervening trajectory for the full allocated regulation band. At the next dispatch update, the calculation repeats from the actual endpoint, incorporating any previous tracking and regulation contribution into the new initial state.

\subsection{Computational Performance}
Computation uses one thread on an AMD EPYC 9V74 processor, with CasADi \cite{casadi}, DAQP \cite{daqp}, ProxQP \cite{proxqp}, HiGHS \cite{highs}, and OSQP \cite{osqp}. Every QP result undergoes independent primal and response-envelope checks. The 39- and 118-bus governor QPs have 18 and 36 variables and 666 and 1062 linear rows, respectively. The common holding-time partition is 0:4:60, 80, 100, 130, 180, 240, 360, 600, 900 s, and infinity.

\begin{figure}[t]
\centering\includegraphics[width=.98\columnwidth]{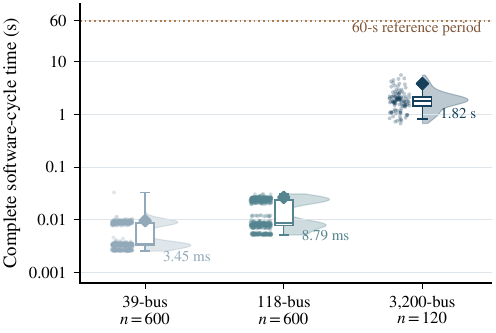}
\caption{Complete software-cycle times on a logarithmic axis: 600 serial calls for each IEEE system and 120 for the synthetic 3,200-bus system. Half violins show the distributions, boxes the interquartile ranges, diamonds the 95th percentiles, and points individual calls. All solver attempts are included.}
\label{fig:timing}
\end{figure}

Fig.~\ref{fig:timing} includes input preparation, constraint assembly, QP solution, reference validation, regulation allocation, AC power flow, and balancing-capability evaluation. The 39-/118-bus medians are 3.45/8.79 ms, 95th percentiles 9.60/26.57 ms, and maxima 33.08/30.84 ms; all 1,200 calls satisfy the tested conditions.

\subsubsection{3,200-bus computational benchmark}
A reproducible synthetic AC system contains 3,200 buses, 1,000 generators at distinct buses, and 3,600 branches. Its connected topology combines a 1,000-bus generator ring, 400 additional interconnections, and 2,200 loaded buses. Total demand is 51.62 GW. The full AC model has 6,398 rectangular voltage variables; the online problem retains all 999 independent generator references plus 999 auxiliaries.

The fixed certificate uses a $0.05\%\,P^{\max}$ injection-domain half-width and 21 bounded load factors, consisting of one common and 20 regional components. Each bus combines its common and regional errors with equal weights, giving a $0.0025\%$ load-error bound. Measurement, regulation, and model-error bounds are $0.0005\%$, $0.0025\%$, and $0.00025\%\,P^{\max}$, respectively. Device response and physical rate limits retain Table~\ref{tab:setup}'s settings. The self-mapping certificate has contraction bound 0.9374. Sparse coefficient preparation and operating-limit construction take 227.55 and 9.90 s, respectively, and the resulting certificate is reused throughout the online tests.

Box-support tests and certified halfspace dominance reduce 247,998 operating/domain inequalities to 2,000 retained rows. Exact scalar elimination of device-wise absolute-value conditions yields a 1,998-variable QP with 2,050 linear rows and 105,894 nonzeros. References are checked against all 2,325 pre-dominance rows. Four seeded request patterns exercise different correction magnitudes in 120 serial calls, including two full AC power-flow checks per call. Median, 95th-percentile, and maximum complete times are \LargeMedian{}, \LargePninetyfive{}, and \LargeMax{} s. Peak online memory is \LargeMemory{} MiB. The observed maximum occupies \LargePercent{}\% of a 60-s reference-update period. \LargeTrajectorySentence{} These measurements establish a practical online computation budget for the fixed-certificate implementation at the tested network and control dimensions.

\section{Conclusion}
This paper constructs convex inner approximations linking AC network constraints, device response, and dispatch-time ramping capability. Fixed-point certification and response envelopes provide affine admissibility conditions. A command governor combines these conditions with headroom, rate limits, and a hard terminal-error bound to select the next target from the actual preceding endpoint. The IEEE 39- and 118-bus experiments complete all 160 primary accepted transitions within five or fifteen minutes while satisfying AC and finite balancing-capability checks. A connected 3,200-bus synthetic system extends the computational assessment to 999 independent references, with a \LargeMedian{}-s median and \LargeMax{}-s maximum complete software cycle. These results support repeated dispatch admission using reusable network certificates and explicit device-response conditions.

\balance
\bibliographystyle{IEEEtran}
\bibliography{references}
\end{document}